\documentclass[11pt]{llncs}
\usepackage{microtype}
\usepackage{makeidx}
\usepackage{graphicx}
\usepackage{xcolor}
\usepackage{ifthen}
\usepackage{microtype}
\usepackage{boxedminipage}
\usepackage{amsmath}
\usepackage{amsfonts}
\usepackage{amssymb}
\usepackage{authblk}
\usepackage{breakcites}
\usepackage{chemscheme}
\usepackage{hyperref}
\usepackage{framed}
\usepackage{float}
\usepackage{breqn}

\newcommand{\sv}[1]{#1}
\newcommand{\lv}[1]{}

\newcommand{\E}{{\bf E}}

\newcommand{\pr}{\mathbf{Pr}}
\newcommand{\eps}{{\varepsilon}}
\newcommand{\etal}{{\it et al. }}

\newcommand{\veps}{\varepsilon}

\newcommand{\cs}[1]{{\mu_{#1}}}

\newcommand{\R}{\mathbb{R}}

\newcommand{\OA}{\mathcal{O}}

\allowdisplaybreaks

\graphicspath{{./Figures/}}

\begin{document}

\title{Approximate Clustering with Same-Cluster Queries}
%
%
\author{Nir Ailon\inst{1} \and Anup Bhattacharya\inst{2} \and Ragesh Jaiswal\inst{2} \and Amit Kumar\inst{2}}
%
%
%
\institute{
Technion, Haifa, Israel.\thanks{Email address: \email{nailon@cs.technion.ac.il}}
\and
Department of Computer Science and Engineering, \\
Indian Institute of Technology Delhi.\thanks{Email addresses: \email{\{anupb, rjaiswal, amitk\}@cse.iitd.ac.in}}
}
\sv{\thispagestyle{empty}}
\maketitle     

\begin{abstract}
Ashtiani \etal proposed a Semi-Supervised Active Clustering framework (SSAC), where the learner is allowed to make adaptive queries to a domain expert. The queries are of the kind ``{\it do two given points belong to the same optimal cluster?}", and the answers to these queries are assumed to be consistent with a unique optimal solution. There are many clustering contexts where such {\em same cluster} queries are feasible. Ashtiani \etal exhibited the power of such queries by showing that any instance of the $k$-means clustering problem, with additional {\em margin} assumption, can be solved efficiently  if one is allowed $O(k^2 \log{k} + k \log{n})$ same-cluster queries. This is interesting since the $k$-means problem, even with the margin assumption, is $\mathsf{NP}$-hard.

In this paper, we extend the work of Ashtiani \etal to the approximation setting showing that a few of such same-cluster queries enables one to get a polynomial-time  $(1 + \veps)$-approximation algorithm for the $k$-means problem without any margin assumption on the input dataset. 
Again, this is interesting since the $k$-means problem is $\mathsf{NP}$-hard to approximate within a factor $(1 + c)$ for a fixed constant $0 < c < 1$. The number of same-cluster queries used is $\textrm{poly}(k/\veps)$ which is independent of the size $n$ of the dataset. 
Our algorithm is based on the $D^2$-sampling technique, also known as the $k$-means++ seeding algorithm. 
We also give a conditional lower bound on the number of same-cluster queries showing that if the Exponential Time Hypothesis (ETH) holds, then any such efficient query algorithm needs to make $\Omega \left(\frac{k}{poly \log k} \right)$ same-cluster queries. 
Our algorithm can be extended for the case when the oracle is faulty, that is, it gives wrong answers to queries with some bounded probability.
Another result we show with respect to the $k$-means++ seeding algorithm is that a small modification to the $k$-means++ seeding algorithm within the SSAC framework converts it to a constant factor approximation algorithm instead of the well known $O(\log{k})$-approximation algorithm.
\end{abstract}

\section{Introduction}
Clustering is extensively used in data mining and is typically the first task performed when trying to understand large data. Clustering basically involves partitioning given data into groups or clusters such that data points within the same cluster are similar as per some similarity measure. Clustering is usually performed in an unsupervised setting where data points do not have any labels associated with them. The partitioning is done using some measure of similarity/dissimilarity between data elements.
In this work, we extend the work of Ashtiani \etal\cite{akd16} who explored the possibility of performing clustering in a {\em semi-supervised active learning} setting for center based clustering problems such as $k$-median/means. In this setting, which they call Semi-Supervised Active Clustering framework or SSAC in short, the clustering algorithm is allowed to make adaptive queries of the form:
\begin{quote}
{\it do two points from the dataset belong to the same optimal cluster?}.
\end{quote} where query answers are assumed to be consistent with a unique optimal solution.
Ashtiani \etal\cite{akd16} started the study of understanding the strength of this model. 
Do hard clustering problems become easy in this model? They explore such questions in the context of center-based clustering problems. Center based clustering problems such as $k$-means are extensively used to analyze large data clustering problems. Let us define the $k$-means problem in the Euclidean setting.

\begin{definition}[$k$-means problem]
Given a dataset $X \subseteq \R^d$ containing $n$ points, and a positive integer $k$, find a set of $k$ points $C \subseteq \R^d$ (called centers) such that the following cost function is minimized:
\[
\Phi(C, X) = \sum_{x \in X} \min_{c \in C} D(x, c).
\]
$D(x, c)$ denotes the squared Euclidean distance between $c$ and $x$. That is, $D(x, c) = ||x - c||^2$.
\end{definition}

Note that the $k$ optimal centers $c_1, ..., c_k$ of the $k$-means problem define $k$ clusters of points in a natural manner. All points for which the closest center is $c_i$ belong to the $i^{th}$ cluster. This is also known as the {\em Voronoi partitioning} and the clusters obtained in this manner using the optimal $k$ centers are called the optimal clusters. Note that the optimal center for the $1$-means problem for any dataset $X \subseteq \R^d$ is the centroid of the dataset denoted by $\mu(X) \stackrel{def.}{=} \frac{\sum_{x \in X} x}{|X|}$.
This means that if $X_1, ...., X_k$ are the optimal clusters for the $k$-means problem on any dataset $X \subseteq \R^d$ and $c_1, ..., c_k$ are the corresponding optimal centers, then $\forall i, c_i = \mu(X_i)$.
The $k$-means problem has been widely studied and various facts are known about this problem.
The problem is tractable when either the number $k$ of clusters or the dimension $d$ equal to $1$.
However, when $k>1$ or $d > 1$, then the problem is known to be
$\mathsf{NP}$-hard~\cite{das08,V09,mnv12}.
There has been a number of works of {\em beyond the worst-case} flavour for $k$-means problem in which it is typically assumed that the dataset satisfies some {\em separation} condition, and then the question is whether this assumption can be exploited to design algorithms providing better guarantees for the problem.
Such questions have led to different definitions of separation and also some interesting results for datasets that satisfy these separation conditions (e.g., \cite{OstrovskyRSS06,BalcanBG09,abs12}).
Ashtiani \etal\cite{akd16} explored one such separation notion that they call the $\gamma$-margin property. 

\begin{definition}[$\gamma$-margin property]
Let $\gamma > 1$ be a real number.
Let $X \subseteq \R^d$ be any dataset and $k$ be any positive integer. 
Let $P_X = \{X_1, ..., X_k\}$ denote $k$ optimal clusters for the $k$-means problem. Then this optimal partition of the dataset $P_X$ is said to satisfy the $\gamma$-margin property iff for all $i \neq j \in \{1, ..., k\}$ and $x \in X_i$ and $y \in X_{j}$, we have:
$$\gamma \cdot ||x - \mu(X_i)|| < ||y - \mu(X_i)||.$$
\end{definition}
Qualitatively, what this means is that every point within some cluster is closer to its own cluster center than a point that does not belong to this cluster. This seems to be a very strong separation property. Ashtiani \etal\cite{akd16} showed that the $k$-means clustering problem is $\mathsf{NP}$-hard even when restricted to instances that satisfy the $\gamma$-margin property for $\gamma = \sqrt{3.4} \approx 1.84$.
Here is the formal statement of their hardness result.

\begin{theorem}[Theorem 10 in \cite{akd16}]
Finding an optimal solution to $k$-means objective function is $\mathsf{NP}$-hard when $k = \Theta(n^{\veps})$ for any $\veps  \in (0, 1)$, even when there is an optimal clustering that satisfies the $\gamma$-margin property for $\gamma = \sqrt{3.4}$.
\end{theorem}

In the context of the $k$-means problem, the {\em same-cluster} queries within the SSAC framework are decision questions of the form: {\it Are points $x,y$ such that $x \neq y$ belong to the same optimal cluster?}
\footnote{In case where the optimal solution is not unique, the same-cluster query answers are consistent with respect to any fixed optimal clustering.}
Following is the main question explored by Ashitiani \etal\cite{akd16}:

\begin{quote}
{\it Under the $\gamma$-margin assumption, for a fixed $\gamma \in (1,\sqrt {3.4}]$, how many queries must be made in the SSAC framework for $k$-means to become tractable?}
\end{quote}

Ashtiani \etal\cite{akd16} addressed the above question and gave a query algorithm.
Their algorithm, in fact, works for a more general setting where the clusters are not necessarily optimal. 
In the more general setting, there is a {\em target} clustering $\bar{X} = \bar{X}_1, ...,\bar{X}_k$ of the given dataset $X \subseteq \R^d$ (not necessarily optimal clusters) such that these clusters satisfy the $\gamma$-margin property (i.e., for all $i, x \in \bar{X}_i$, and $y \notin \bar{X}_i, \gamma \cdot ||x - \mu(\bar{X}_i)|| < ||y - \mu(\bar{X}_i)||$) and the goal of the query algorithm is to output the clustering $\bar{X}$. Here is the main result of Ashtiani \etal

\begin{theorem}[Theorems 7 and 8 in \cite{akd16}]\label{thm:ash-main}
Let $\delta \in (0, 1)$ and $\gamma > 1$. Let $X \subseteq \R^d$ be any dataset containing $n$ points, $k$ be a positive integer, and $X_1,...,X_k$ be any target clustering of $X$ that satisfies the $\gamma$-margin property. Then there is a query algorithm $A$ that makes $O\left(k \log{n} + k^2\frac{\log{k} + \log{1/\delta}}{(\gamma - 1)^4} \right)$ same-cluster queries and with probability at least $(1 - \delta)$ outputs the clustering $X_1,...,X_k$. The running time of algorithm $A$ is $O\left(k n\log{n} + k^2\frac{\log{k} + \log{1/\delta}}{(\gamma - 1)^4} \right)$.
\end{theorem}

The above result is a witness to the power of the SSAC framework.
We extend this line of work by examining the power of same-cluster queries in the {\em approximation algorithms} domain.
Our results do not assume any separation condition on the dataset (such as $\gamma$-margin as in \cite{akd16}) and they hold for {\em any} dataset.

Since the $k$-means problem is $\mathsf{NP}$-hard, an important line of research work has been to obtain approximation algorithms for the problem.
There are various efficient approximation algorithms for the $k$-means problem, the current best approximation guarantee being $6.357$ by Ahmadian \etal\cite{ANSW2016}.
A simple approximation algorithm that gives an $O(\log{k})$ approximation guarantee is the $k$-means++ seeding algorithm (also known as $D^2$-sampling algorithm) by Arthur and Vassilvitskii~\cite{ArthurV07}.
This algorithm is commonly used in solving the $k$-means problem in practice.
As far as {\em approximation schemes} or in other words $(1 + \veps)$-approximation algorithms (for arbitrary small $\veps < 1$) are concerned, the following is known:
It was shown by Awasthi \etal\cite{acks15} that there is some fixed constant $0 < c <1$ such that there cannot exist an efficient $(1 + c)$ factor approximation unless $\mathsf{P} = \mathsf{NP}$.
This result was improved by Lee \etal\cite{lsw17} where it was shown that it is $\mathsf{NP}$-hard to approximate the $k$-means problem within a factor of $1.0013$. 
However, when either $k$ or $d$ is a fixed constant, then there are Polynomial Time Approximation Schemes (PTAS) for the $k$-means problem.\footnote{This basically means an algorithm that runs in time polynomial in the input parameters but may run in time exponential in $1/\veps$.}
Addad \etal\cite{addad16} and Friggstad \etal\cite{friggstad16} gave a PTAS for the $k$-means problem in constant dimension.
For fixed constant $k$, various PTASs are known~\cite{kss,FeldmanMS07,jks,jky15}.
Following is the main question that we explore in this work:
\begin{quote}
{\it For arbitrary small $\veps > 0$, how many same-cluster queries must be made in an efficient $(1+\veps)$-approximation algorithm for $k$-means in the SACC framework? The running time should be polynomial in all input parameters such as $n, k, d$ and also in $1/\veps$.}
\end{quote}

Note that this is a natural extension of the main question explored by Ashtiani \etal\cite{akd16}. 
Moreover, we have removed the separation assumption on the data. 
We provide an algorithm that makes $\textrm{poly}(k/\veps)$ same-cluster queries and runs in time $O(nd \cdot \textrm{poly}(k/\veps))$. 
More specifically, here is the formal statement of our main result:

\begin{theorem}[Main result: query algorithm]\label{thm:main1}
Let $0 < \veps \leq 1/2$, $k$ be any positive integer, and $X \subseteq \R^d$. Then there is a query algorithm {\tt A} that runs in time $O(nd k^{9}/\veps^4)$ and with probability at least $0.99$ outputs a center set $C$ such that $\Phi(C, X) \leq (1 + \veps) \cdot \Delta_k(X)$. Moreover, the number of same-cluster queries used by {\tt A} is $O(k^{9}/\veps^4)$. 
Here $\Delta_k(X)$ denotes the optimal value of the $k$-means objective function.
\end{theorem}

Note that unlike Theorem~\ref{thm:ash-main}, our bound on the number of same-cluster queries is independent of the size of the dataset. We find this interesting and the next natural question we ask is whether this bound on the number of same-cluster queries is {\em tight} in some sense. In other words, does there exist a query algorithm in the SSAC setting that gives $(1 + \veps)$-approximation in time polynomial in $n, k, d$ and makes significantly fewer queries than the one given in the theorem above?
We answer this question in the negative by establishing a conditional lower bound on the number of same-cluster queries under the assumption that ETH (Exponential Time Hypothesis) \cite{IP01,IPZ01} holds.
The formal statement of our result is given below.

\begin{theorem}[Main result: query lower bound]\label{thm:main2}
If the Exponential Time Hypothesis (ETH) holds, then there exists a constant $c>1$ such that any $c$-approximation query algorithm for the $k$-means problem that runs in time $\textrm{poly}(n, d, k)$ makes at least $\frac{k}{\textrm{poly}\log k}$ queries.
\end{theorem}

\paragraph{Faulty query setting}
The existence of a same-cluster oracle that answers such queries perfectly may be too strong an assumption.
A more reasonable assumption is the existence of a {\em faulty} oracle that can answer incorrectly but only with bounded probability.
Our query approximation algorithm can be extended to the setting where answers to the same-cluster queries are {\em faulty}.
More specifically, we can get wrong answers to queries independently but with probability at most some constant $q<1/2$. 
Also note that in our model the answer for a same-cluster query does not change with repetition. This means that one cannot ask the same query multiple times and amplify the probability of correctness.
We obtain $(1+\eps)$-approximation guarantee for the $k$-means clustering problem in this setting.
The main result is given as follows.

\begin{theorem} \label{thm:faulty} 
Let $0<\eps\leq 1/2$, $k$ be any positive integer, and $X\subseteq \R^d$. 
Consider a faulty SSAC setting where the response to every same-cluster query is incorrect with probability at most some constant $q < 1/2$.
In such a setting, there is a query algorithm {\tt $A^E$} that 
with probability at least $0.99$, outputs a center set $C$ such that $\Phi(C,X)\leq (1+\veps) \cdot \Delta_k(X)$. Moreover, the number of same-cluster queries used by {\tt $A^E$} is $O(k^{15}/\veps^8)$. 
\end{theorem}

The previous theorems summarise the main results of this work which basically explores the power of same-cluster queries in designing fast $(1+\veps)$-approximation algorithms for the $k$-means problem.
We will give the proofs of the above theorems in Sections~\ref{sec:main1}, \ref{sec:main2}, and \ref{sec:faulty}.
There are some other simple and useful contexts, where the SSAC framework gives extremely nice results.
One such context is the popular $k$-means++ seeding algorithm.
This is an extremely simple sampling based algorithm for the $k$-means problem that samples $k$ centers in a sequence of $k$ iterations.
We show that within the SSAC framework, a small modification of this sampling algorithm converts it to one that gives constant factor approximation instead of $O(\log{k})$-approximation~\cite{ArthurV07} that is known.
This is another witness to the power of same-cluster queries.
We begin the technical part of this work by discussing this result in Section~\ref{sec:kmpp}.
Some of the basic techniques involved in proving our main results will be introduced while discussing this simpler context.

\paragraph{Other related work}
Clustering problems have been studied in different semi-supervised settings. Basu \etal\cite{BBM2004} explored {\em must-link} and {\em cannot-link} constraints in their semi-supervised clustering formulation. 
In their framework, must-link and cannot-link constraints were provided explicitly as part of the input along with the cost of violating these constraints. 
They gave an active learning formulation for clustering in which an oracle answers whether two query points belong to the same cluster or not, and gave a clustering algorithm using these queries. 
However, they work with a different objective function and there is no discussion on theoretical bounds on the number of queries.
In contrast, in our work we consider the $k$-means objective function and provide bounds on approximation guarantee, required number of adaptive queries, and the running time.
Balcan and Blum \cite{BB2008} proposed an interactive framework for clustering with {\em split/merge} queries. 
Given a clustering $C=\{C_1,\ldots\}$, a user provides feedback by specifying that some cluster $C_l$ should be split, or clusters $C_i$ and $C_j$ should be merged. 
Awasthi \etal\cite{ABV2014} studied a local interactive algorithm for clustering with split and merge feedbacks.
Voevodski \etal\cite{VBRTX2014} considered \textit{one versus all} queries where query answer for a point $s \in X$ returns distances between $s$ to all points in $X$. 
For a $k$-median instance satisfying $(c,\eps)$-approximation stability property \cite{BalcanBG09}, the authors found a clustering close to true clustering using only $O(k)$ one versus all queries.
Vikram and Dasgupta \cite{VD16} designed an interactive bayesian hierarchical clustering algorithm. 
Given dataset $X$, the algorithm starts with a candidate hierarchy $T$, and an initially empty set $C$ of constraints. 
The algorithm queries user with a subtree $T|_S$ of hierarchy $T$ restricted to constant sized set $S \subset X$ of leaves. 
User either accepts $T|_S$ or provides a counterexample triplet $(\{a,b\},c)$ which the algorithm adds to its set of constraints $C$, and updates $T$. 
They consider both random and adaptive ways to select $S$ to query $T|_S$.

\paragraph{Our Techniques}
We now give a brief outline of the new ideas needed for our results. Many algorithms for the $k$-means problem proceed by iteratively finding approximations to the optimal centers. 
One such popular algorithm is the $k$-means++ seeding algorithm~\cite{ArthurV07}. In this algorithm, one builds a set of potential centers iteratively. 
We start with a set $C$ initialized to the empty set. 
At each step, we sample a point with probability proportional to the square of the distance from $C$, and add it to $C$. 
Arthur and Vassilvitskii~\cite{ArthurV07} showed that if we continue this process till $|C|$ reaches 
$k$, then the corresponding $k$-means solution has expected cost $O(\log{k})$ times the optimal $k$-means cost.
Aggarwal \etal\cite{AggarwalDK09} showed that if we continue this process till $|C|$ reaches $\beta k$, for some constant $\beta > 1$, then the corresponding $k$-means solution (where we actually open all the centers in $C$) has cost which is within constant factor of the optimal $k$-means cost with high probability. 
Ideally, one would like to stop when size of $C$ reaches $k$ and obtain a constant factor approximation guarantee.
We know from previous works~\cite{ArthurV07,BR13,BJN16} that this is not possible in the classical (unsupervised) setting.
In this work, we show that one can get such a result in the SSAC framework. 
A high-level way of analysing the $k$-means++ seeding algorithm is as follows.
We first observe that if we randomly sample a point from a cluster, then the expected  cost of assigning all points of this cluster to the sampled point is within a constant factor of the cost of assigning all the points to the mean of this cluster. Therefore, it suffices to select a point chosen uniformly at random from each of the clusters. Suppose the set $C$ contains such samples for the 
first $i$ clusters (of an optimal solution). If the other clusters are far from these $i$ clusters, then it is likely that the next point added to $C$ belongs to a new cluster (and perhaps is close to a uniform sample). 
However to make this more probable, one needs to add several points to $C$. 
Further, the number of samples that needs to be added to $C$ starts getting worse as the value of $i$ increases. Therefore, the algorithm needs to build $C$ till its size becomes $O(k \log k)$. In the SSAC framework, we can {\em tell} if the next point added in $C$ belongs to a new cluster or not. Therefore, we can always ensure that $|C|$ does not exceed $k$. To make this idea work, we need to extend the induction argument of Arthur and Vassilvitskii~\cite{ArthurV07} -- details are given in Section~\ref{sec:kmpp}.

We now explain the ideas for the PTAS for $k$-means. We consider the special case of $k=2$. 
Let $X_1$ and $X_2$ denote the optimal clusters with $X_1$ being the larger cluster. 
Inaba \etal\cite{inaba} showed that if we randomly sample about $O(1/\veps)$ points from a cluster, and let $\mu'$ denote the mean of this subset of sampled points, then the cost of assigning 
all points in the cluster to $\mu'$ is within $(1+\veps)$ of the cost of assigning all these points to their actual mean (whp).
Therefore, it is enough to get uniform samples of size about $O(1/\veps)$ from each of the clusters. 
Jaiswal \etal\cite{jks} had the following approach for obtaining a $(1+ \veps)$-approximation algorithm
for $k$-means (with running time being $nd \cdot f(k,\veps)$, where $f$ is an exponential function of $k/\veps$) -- suppose we sample   about $O(1/\veps^2)$ points from the input, call this sample $S$. 
It is likely to contain at least $O(1/\veps)$  from $X_1$, but we do not know which points in $S$ are from $X_1$. 
Jaiswal \etal addressed this problem by cycling over all subsets of $S$. 
In the SSAC model, we can directly partition $S$ into $S \cap X_1$ and $ S \cap X_2$ using $|S|$ same-cluster queries. 
Having obtained such a sample $S$, we can get a close approximation to the mean of $X_1$. 
So assume for sake of simplicity that we know $\mu_1$, the mean of $X_1$. 
Now we are faced with the problem of obtaining a uniform sample from $X_2$.  
The next idea of Jaiswal \etal is to sample points with probability proportional to square of distance from $\mu_1$.
This is known as $D^2$-sampling.
Suppose we again sample about $O(1/\veps^2)$ such points, call this sample $S'$. 
Assuming that the two clusters are far enough (otherwise the problem only gets easier), they show that $S'$ will contain about $O(1/\veps^2)$ points from $X_2$ (with good probability). 
Again, in the SSAC model, we can find this subset by $|S'|$ queries -- call this set $S''$. 
However, the problem is that $S''$ may not represent a uniform sample from $X_2$. 
For any point $e \in X_2$, let $p_e$ denote the conditional probability of sampling $e$ given that a point from $X_2$ is sampled when sampled using $D^2$-sampling.
They showed $p_e$ is at least  $\frac{\veps}{m}$, where $m$ denotes the size of $X_2$. 
In order for the sampling lemma of Inaba \etal\cite{inaba} to work, we cannot work with approximately uniform sampling. 
The final trick of Jaiswal \etal was to show that one can in fact get a uniform sample of size about $O(\veps |S''|) = O(1/\veps)$ from $S''$. 
The idea is as follows -- for every element $e \in S''$, we {\em retain} it with probability $\frac{\veps}{p_e m}$ (which is at most $1$), otherwise we remove it from $S''$. It is not difficult to see that this gives a uniform sample from $X_2$. 
The issue is that we do not know $m$. 
Jaiswal \etal again cycle over all subsets of $S'$ -- we know that there is a (large enough) subset of $S'$ which will behave like a uniform sample from $X_2$. 
In the SSAC framework, we first identify the subset of $S'$ which belongs to $X_2$, call this $S''$ (as above). 
Now we prune some points from $S''$ such that the remaining points behave like a uniform sample. This step is non-trivial because as indicated above, we do not know the value $m$. 
Instead, we first show that $p_e$ lies between $\frac{\veps}{m}$ and $\frac{2}{m}$ for most of the points of $X_2$. Therefore, $S''$ is likely  to contain such a {\em nice} point, call it $v$. 
Now, for every point $e \in S''$, we retain it with probability $\frac{\veps p_e}{2p_v}$ (which we know is at most $1$). 
This gives a uniform sample of sufficiently large size from $X_2$. 
For $k$ larger than $2$, we generalize the above ideas using a non-trivial induction argument.

\section{$k$-means++ within SSAC framework}\label{sec:kmpp}
The $k$-means++ seeding algorithm, also known as the $D^2$-sampling algorithm, is a simple sampling procedure that samples $k$ centers in $k$ iterations. The description of this algorithm is given below.

\begin{table}[h]
\centering
\begin{tabular}{| l || l |}
\hline
{\tt $k$-means++($X$,$k$)} & {\tt Query-$k$-means++($X, k$)} \\
\hspace{0.1in} - Randomly sample a point $x \in X$ & \hspace{0.1in} - Randomly sample a point $x \in X$\\
\hspace{0.1in} - $C \leftarrow \{x\}$ & \hspace{0.1in} - $C \leftarrow \{x\}$ \\
\hspace{0.1in} - for $i$ = $2$ to $k$ & \hspace{0.1in} - for $i$ = $2$ to $k$ \\
\hspace{0.3in} - Sample $x \in X$ using distribution $D$ & \hspace{0.3in} - for $j$ = $1$ to $\lceil \log{k} \rceil$\\
\hspace{0.4in} defined as $D(x) = \frac{\Phi(C, \{x\})}{\Phi(C, X)}$ & \hspace{0.5in} - Sample $x \in X$ using distribution $D$\\
\hspace{0.3in} - $C \leftarrow C \cup \{x\}$ & \hspace{0.6in}  defined as $D(x) = \frac{\Phi(C, \{x\})}{\Phi(C, X)}$\\
\hspace{0.1in} - return($C$) & \hspace{0.5in} - if({\tt NewCluster($C, x$)})\{$C \leftarrow C\cup \{x\}$; break\}\\
 & \hspace{0.1in} - return($C$)\\\cline{2-2}
 & {\tt NewCluster($C, x$)}\\
 & \hspace{0.1in} - If($\exists c \in C$ s.t. {\tt SameCluster($c, x$)}) return(false)\\
 & \hspace{0.1in} - else return(true)\\
\hline
\end{tabular}
\caption{$k$-means++ seeding algorithm (left) and its adaptation in the SSAC setting (right)}
\label{table:1}
\end{table}

The algorithm picks the first center randomly from the set $X$ of points and after having picked the first $(i-1)$ centers denoted by $C_{i-1}$, it picks a point $x \in X$ to be the $i^{\text{th}}$ center with probability proportional to $\min_{c \in C_{i-1}} ||x - c||^2$. The running time of $k$-means++ seeding algorithm is clearly $O(nkd)$. Arthur and Vassilvitskii~\cite{ArthurV07} showed that this simple sampling procedure gives an $O(\log{k})$ approximation in expectation for any dataset. Within the SSAC framework where the algorithm is allowed to make same-cluster queries, we can make a tiny addition to the $k$-means++ seeding algorithm to obtain a query algorithm that gives constant factor approximation guarantee and makes only $O(k^2 \log{k})$ same-cluster queries. The description of the query algorithm is given in Table~\ref{table:1} (see right). 
In iteration $i>1$, instead of simply accepting the sampled point $x$ as the $i^{\text{th}}$ center (as done in $k$-means++ seeding algorithm), the sampled point $x$ is accepted only if it belongs to a cluster other than those to which centers in $C_{i-1}$ belong (if this does not happen, the sampling is repeated for at most $\lceil \log{k} \rceil$ times). Here is the main result that we show for the {\tt query-$k$-means++} algorithm.

\begin{theorem}\label{thm:kmpp}
Let $X \subseteq \R^d$ be any dataset containing $n$ points and $k>1$ be a positive integer.
Let $C$ denote the output of the algorithm {\tt Query-$k$-means++($X, k$)}. Then
\[
\E[\Phi(C, X)] \leq 24 \cdot \Delta_k(X),
\]
where $\Delta_k(X)$ denotes the optimal cost for this $k$-means instance.
Furthermore, the algorithm makes $O(k^2 \log{k})$ same-cluster queries and the running time of the algorithm is $O(nkd + k \log{k} \log{n} + k^2 \log{k})$.
\end{theorem}
The bound on the number of same-cluster queries is trivial from the algorithm description. 
For the running time, it takes $O(nd)$ time to update the distribution $D$ which is updated $k$ times. 
This accounts for the $O(nkd)$ term in the running time. 
Sampling an element from a distribution $D$ takes $O(\log{n})$ time (if we maintain the cumulative distribution etc.) and at most $O(k \log{k})$ points are sampled. 
Moreover, determining whether a sampled point belongs to an uncovered cluster takes $O(k)$ time. 
So, the overall running time of the algorithm is $O(nkd + k \log{k} \log{n} + k^2 \log{k})$.
We prove the approximation guarantee in the remaining discussion. We will use the following terminology. Let the optimal $k$ clusters for dataset $X$ are given as $X_1, ..., X_k$. For any $i$, $\Delta_i(X)$ denotes the optimal cost of the $i$-means problem on dataset $X$. Given this, note that $\Delta_k(X) = \sum_{i=1}^{k} \Delta_1(X_i)$. For any non-empty center set $C$, we say that a point $x$ is sampled from dataset $X$ using $D^2$-sampling w.r.t. center set $C$ if the sampling probability of $x \in X$ is given by $D(x) = \frac{\Phi(C, \{x\})}{\Phi(C, X)}$.

The proof of Theorem~\ref{thm:kmpp} will mostly follow $O(\log{k})$-approximation guarantee proof of $k$-means++ seeding by Arthur and Vassilvitskii~\cite{ArthurV07}. The next two lemmas from \cite{ArthurV07} are crucially used in the proof of approximation guarantee.

\begin{lemma}\label{lemma:av1}
Let $A$ be any optimal cluster and let $c$ denote a point sampled uniformly at random from $A$. Then
$\E[\Phi(\{c\}, A)] \leq 2 \cdot \Delta_1(A)$.
\end{lemma}

\begin{lemma}\label{lemma:av2}
Let $C$ be any arbitrary set of centers and let $A$ be any optimal cluster. Let $c$ be a point sampled with $D^2$-sampling with respect to the center set $C$. Then $\E[\Phi(C \cup \{c\}, A) | c \in A] \leq 8 \cdot \Delta_1(A)$.
\end{lemma}

The first lemma says that a randomly sampled center from $X$ provides a good approximation (in expectation) to the cost of the cluster to which it belongs. The second lemma says that for any center set $C$, given that a center $c$ that is $D^2$-sampled from $X$ w.r.t. $C$ belong to an optimal cluster $A$, the conditional expectation of the cost of the cluster $A$ with respect to center set $C \cup \{c\}$ is at most $8$ times the optimal cost of cluster $A$. Using the above two lemmas, let us try to qualitatively see why the $k$-means++ seeding algorithm behave well. The first center belongs to some optimal cluster $A$ and from Lemma~\ref{lemma:av1} we know that this center is good for this cluster. At the time the $i^{\text{th}}$ center is $D^2$-sampled, there may be some optimal clusters which are still costly with respect to the center set $C_{i-1}$. But then we can argue that it is likely that the $i^{\text{th}}$ sampled center $c$ will belong to one of these costly clusters, and conditioned on the center being from one such cluster $A$, the cost of this cluster after adding $c$ to the current center set is bounded using Lemma~\ref{lemma:av2}. The formal proof of $O(\log{k})$ approximation guarantee in \cite{ArthurV07} involves setting up a clever induction argument. We give a similar induction based argument to prove Theorem~\ref{thm:kmpp}. We prove the following lemma (similar to Lemma 3.3 in \cite{ArthurV07}). We will need the following definitions: 
For any center set $C$, an optimal cluster $A$ is said to be ``covered" if at least one point from $A$ is in $C$, otherwise $A$ is said to be ``uncovered". 
Let $T$ be a union of a subset of the optimal clusters, then we will use the notation $\Phi_{OPT}(T) \stackrel{def.}{=} \sum_{X_i \subseteq T} \Delta_1(X_i)$.

\begin{lemma}\label{lemma:av-induction}
Let $C \subseteq X$ be any set of centers such that the number of uncovered clusters w.r.t. $C$ is $u>0$. Let $X_u$ denote the set of points of the uncovered clusters and $X_c$ denote set of the points of the covered clusters. Let us run $t$ iterations of the outer for-loop in {\tt Query-$k$-means++} algorithm such that $t \leq u \leq k$. Let $C'$ denote the resulting set of centers after running $t$ iterations of the outer for-loop. Then the following holds:
\begin{equation}\label{eqn:induction}
\E[\Phi(C', X)] \leq (\Phi(C, X_c) + 8 \cdot \Phi_{OPT}(X_u)) \cdot \left(2 + \frac{t}{k} \right) + \frac{u - t}{u} \cdot \Phi(C, X_u).
\end{equation}
\end{lemma}

\begin{proof}
Let us begin by analysing what happens when starting with $C$, one iteration of the outer for-loop in {\tt query-$k$-means++} is executed. The following two observations will be used in the induction argument:\\
\underline{{\it Observation 1}}: If $\frac{\Phi(C, X_c)}{\Phi(C, X)} \geq 1/2$, then we have $\frac{\Phi(C, X_c)}{\Phi(C, X_c) + \Phi(C, X_u)} \geq 1/2$ which implies that $\Phi(C, X_u) \leq \Phi(C, X_c)$, and also $\Phi(C, X) \leq 2 \cdot \Phi(C, X_c)$.\\
\underline{{\it Observation 2}}: If $\frac{\Phi(C, X_c)}{\Phi(C, X)} < 1/2$, then the probability that no point will be added after one iteration is given by $\left(\frac{\Phi(C, X_c)}{\Phi(C, X)} \right)^{\lceil \log{k} \rceil}  < \left(\frac{1}{2}\right)^{\log{k}} = \frac{1}{k}$.

We will now proceed by induction.
We show that if the statement holds for $(t-1, u)$ and $(t-1, u-1)$, then the statement holds for $(t, u)$.
In the basis step, we will show that the statement holds for $t=0$ and $u > 0$ and $u = t = 1$.\\
\underline{\bf Basis step}: Let us first prove the simple case of $t=0$ and $u > 0$. In this case, $C' = C$.
So, we have $\E[\Phi(C', X)] = \Phi(C, X)$ which is at most the RHS of (\ref{eqn:induction}).
Consider the case when $u = t = 1$. This means that there is one uncovered cluster and one iteration of the outer for-loop is executed. If a center from the uncovered cluster is added, then $\E[\Phi(C', X)] \leq \Phi(C, X_c) + 8 \cdot \Phi_{OPT}(X_u)$ and if no center is picked, then $\Phi(C', X) = \Phi(C, X)$.
The probability of adding a center from the uncovered cluster is given by $p = 1 - \left( \frac{\Phi(C, X_c)}{\Phi(C, X)} \right)^{\log{k}}$. So, we get $\E[\Phi(C', X)] \leq p \cdot (\Phi(C, X_c) + 8 \cdot \Phi_{OPT}(X_u)) + (1-p) \cdot \Phi(C, X)$. Note that this is upper bounded by the RHS of (\ref{eqn:induction}) by observing that $1 - p \leq \frac{\Phi(C, X_c)}{\Phi(C, X)}$.\\
\underline{\bf Inductive step}: As stated earlier, we will assume that the statement holds for $(t-1, u)$ and $(t-1, u-1)$ and we will show that the statement holds for $(t,u)$. Suppose $p \stackrel{def.}{=} \frac{\Phi(C, X_c)}{\Phi(C, X)} \geq \frac{1}{2}$, then $\Phi(C, X) \leq 2 \cdot \Phi(C, X_c)$ and so $\Phi(C', X) \leq \Phi(C,X) \leq 2 \cdot \Phi(C, X_c)$ which is upper bounded by the RHS of (\ref{eqn:induction}). So, the statement holds for $(t,u)$ (without even using the induction assumption). So, for the rest of the discussion, we will assume that $p < 1/2$. Let us break the remaining analysis into two cases -- (i) no center is added in the next iteration of the outer for-loop, and (ii) a center is added. In case (i), $u$ does not change, $t$ decreases by $1$, and the covered and uncovered clusters remain the same after the iteration. So the contribution of this case to $\E[\Phi(C', X)]$ is at most
\begin{equation}\label{eqn:ind-1}
p^{\lceil \log{k} \rceil} \cdot \left( (\Phi(C, X_c) + 8 \cdot \Phi_{OPT}(X_u)) \cdot \left(2 + \frac{t - 1}{k} \right) + \frac{u - t + 1}{u} \cdot \Phi(C, X_u)\right)
\end{equation}
Now, consider case (ii). Let $A$ be any uncovered cluster w.r.t. center set $C$. For any point $a \in A$, let $p_a$ denote the conditional probability of sampling $a$ conditioned on sampling a point from $A$. Also, let $\phi_a$ denote the cost of $A$ given $a$ is added as a center. That is, $\phi_a = \Phi(C \cup \{a\}, A)$.
The contribution of $A$ to the expectation $\E[\Phi(C', X)]$ using the induction hypothesis is:
\begin{dmath*}
(1 - p^{\lceil \log{k} \rceil}) \cdot \frac{\Phi(C, A)}{\Phi(C, X_u)} \sum_{a \in A} p_a \cdot \left( \left(\Phi(C, X_c) + \phi_a  + 8 \cdot \Phi_{OPT}(X_u) - 8 \cdot \Delta_1(A)\right) \cdot \left(2 + \frac{t-1}{k}\right) +
\frac{u-t}{u-1} \cdot (\Phi(C, X_u) - \Phi(C, A))\right)
\end{dmath*}
This is at most
\begin{dmath*}
(1 - p^{\lceil \log{k}\rceil}) \cdot \frac{\Phi(C, A)}{\Phi(C, X_u)}  \left( \left(\Phi(C, X_c)  + 8 \cdot \Phi_{OPT}(X_u) \right) \cdot \left(2 + \frac{t-1}{k}\right) +
\frac{u-t}{u-1} \cdot (\Phi(C, X_u) - \Phi(C, A))\right)
\end{dmath*}
The previous inequality follows from the fact that $\sum_{a \in A} p_a \phi_a \leq 8 \cdot \Delta_1(A)$ from Lemma~\ref{lemma:av2}. Summing over all uncovered clusters, the overall contribution in case (ii) is at most:
\begin{dmath*}
(1 - p^{\lceil \log{k} \rceil}) \cdot  \left( \left(\Phi(C, X_c)  + 8 \cdot \Phi_{OPT}(X_u) \right) \cdot \left(2 + \frac{t-1}{k}\right) +
\frac{u-t}{u-1} \cdot \left(\Phi(C, X_u) - \frac{\Phi(C, X_u)}{u} \right)\right)
\end{dmath*}
The above bound is obtained using the fact that $\sum_{A \textrm{ is uncovered}} \Phi(C, A)^2 \geq \frac{1}{u} \Phi(C, X_u)^2$. So the contribution is at most
\begin{dmath}\label{eqn:ind-2}
(1 - p^{\lceil \log{k} \rceil}) \cdot  \left( \left(\Phi(C, X_c)  + 8 \cdot \Phi_{OPT}(X_u) \right) \cdot \left(2 + \frac{t-1}{k}\right) +  \frac{u-t}{u} \cdot \Phi(C, X_u) \right)
\end{dmath}
Combining inequalities (\ref{eqn:ind-1}) and (\ref{eqn:ind-2}), we get the following:
\begin{eqnarray*}
\E[\Phi(C', X)] &\leq& \left( \left(\Phi(C, X_c)  + 8 \cdot \Phi_{OPT}(X_u) \right) \cdot \left(2 + \frac{t-1}{k}\right) +  \frac{u-t}{u} \cdot \Phi(C, X_u) \right) + p^{\lceil \log{k} \rceil} \cdot \frac{\Phi(C, X_u)}{u} \\
&=& \left( \left(\Phi(C, X_c)  + 8 \cdot \Phi_{OPT}(X_u) \right) \cdot \left(2 + \frac{t-1}{k}\right) +  \frac{u-t}{u} \cdot \Phi(C, X_u) \right) + \\
&& \left( \frac{\Phi(C, X_c)}{\Phi(C, X)}\right)^{\lceil \log{k} \rceil} \cdot \frac{\Phi(C, X_u)}{u}\\
&\leq& \left( \left(\Phi(C, X_c)  + 8 \cdot \Phi_{OPT}(X_u) \right) \cdot \left(2 + \frac{t-1}{k}\right) +  \frac{u-t}{u} \cdot \Phi(C, X_u) \right) +  \frac{\Phi(C, X_c)}{k u} \\
&& \textrm{(using the Observation 2, that is $p^{\lceil \log{k} \rceil} \leq 1/k$)}\\
&\leq& \left(\Phi(C, X_c)  + 8 \cdot \Phi_{OPT}(X_u) \right) \cdot \left(2 + \frac{t}{k}\right) +  \frac{u-t}{u} \cdot \Phi(C, X_u)
\end{eqnarray*}
This completes the inductive argument and the proof.\qed
\end{proof}

Let us now conclude the proof of Theorem~\ref{thm:kmpp} using the above lemma. Consider the center set $C$ before entering the outer for-loop. This contains a single center $c$ chosen randomly from the dataset $X$. Let $c$ belong to some optimal cluster $A$. Let $C'$ denote the center set after the execution of the outer for-loop completes. Applying the above lemma with $u = t = k-1$, we get:
\begin{eqnarray*}
\E[\Phi(C', X)] &\leq& (\Phi(C, A) + 8 \cdot \Delta_k(X) - 8 \cdot \Delta_1(A)) \cdot \left( 2 + \frac{k-1}{k}\right) \\
&\leq& 3 \cdot (2 \cdot \Delta_1(A) + 8 \cdot \Delta_k(X) - 8 \cdot \Delta_1(A)) \quad \textrm{(using Lemma~\ref{lemma:av1})}\\
&\leq& 24 \cdot \Delta_k(X)
\end{eqnarray*}

\section{Query Approximation Algorithm (proof of Theorem~\ref{thm:main1})}\label{sec:main1}
As mentioned in the introduction, our query algorithm is based on the $D^2$-sampling based algorithm of Jaiswal \etal\cite{jks,jky15}. 
The algorithm in these works give a $(1+\veps)$-factor approximation for arbitrary small $\veps > 0$. 
The running time of these algorithms are of the form $nd \cdot f(k, \veps)$, where $f$ is an exponential function of $k/\veps$. 
We now show that it is possible to get a running time which is polynomial in $n,k,d, 1/\veps$ in the SSAC model. 
The main ingredient in the design and analysis of the sampling algorithm is the following lemma by Inaba \etal\cite{inaba}.

\begin{lemma}[\cite{inaba}]\label{lemma:inaba}
Let $S$ be a set of points obtained by independently sampling $M$ points uniformly at random with replacement from a point set $X \subset \mathbb{R}^d$. Then for any $\delta > 0$,
\[
\pr \left[ \Phi(\{\mu(S)\}, X) \leq \left(1 + \frac{1}{\delta M} \right) \cdot \Delta_1(X)\right] \geq (1 - \delta).
\]
Here $\mu(S)$ denotes the geometric centroid of the set $S$.
That is $\mu(S) = \frac{\sum_{s \in S} s}{|S|}$
\end{lemma}

Our algorithm {\tt Query-$k$-means} is described in Table~\ref{fig:algo}. It maintains a set $C$ of potential centers of the clusters. In each iteration of step~{\bf(3)}, it adds one more candidate center to the set $C$ (whp), and so, the algorithm stops when $|C|$ reaches $k$. 
For sake of explanation, assume  that optimal clusters are $X_1, X_2, \ldots, X_k$ with means $\mu_1, \ldots, \mu_k$ respectively. 
Consider the $i^{\textrm{th}}$ iteration of step~{\bf(3)}. 
At this time $|C|=i-1$, and it has good approximations to means of $i-1$ clusters among $X_1, \ldots. X_k$.
Let us call these clusters {\em covered}.
In Step~{\bf(3.1)}, it samples $N$ points, each with probability proportional to square of distance from $C$ ($D^2$-sampling). 
Now, it partitions this set, $S$, into $S\cap X_1, \ldots, S \cap X_k$ in the procedure {\tt UncoveredClusters}, and then picks the partition with the largest size such that the corresponding optimal cluster $X_j$ is not one of the $(i-1)$ covered clusters.
Now, we would like to get a uniform sample from $X_j$ -- recall that $S \cap X_j$ does not represent a uniform sample. However, as mentioned in the introduction, we need to find an element $s$ of $X_j$ for which the probability of being in sampled is small enough. Therefore, we pick the element in $S \cap X_j$ for which this probability is smallest (and we will show that it has the desired properties). The procedure~{\tt UncoveredCluster} returns this element $s$. Finally, we choose a subset $T$ of $S \cap X_j$ in the procedure {\tt UniformSample}. 
This procedure is designed such that each element of $X_j$ has the same probability of being in $T$. 
In step {\bf (3.4)}, we check whether the multi-set $T$ is of a desired minimum size. 
We will argue that the probability of $T$ not containing sufficient number of points is very small.
If we have $T$ of the desired size, we take its mean and add it to $C$ in Step~{\bf (3.6)}.

\begin{table}
\centering
\begin{tabular}{| l | l |}
\hline
{\bf Constants}: $N = \frac{(2^{12})k^3}{\eps^2}$, $M = \frac{64k}{\veps}$, $L = \frac{(2^{23})k^2}{\veps^4}$ & \\
\cline{1-1}
{\tt Query-$k$-means($X, k, \veps$)} &
{\tt UncoveredCluster($C, S, R$)}\\ 
\hspace*{0.1in} {\bf (1)} $R \leftarrow \emptyset$ &
\hspace*{0.1in} - For all $i \in \{1, ..., k\}$: $S_i \leftarrow \emptyset$\\
\hspace*{0.1in} {\bf (2)} $C \leftarrow \emptyset$ &
\hspace*{0.1in} - $i \leftarrow 1$ \\
\hspace*{0.1in} {\bf (3)} for $i$ = $1$ to $k$  &
\hspace*{0.1in} - For all $y \in R$: \{$S_i \leftarrow y$; $i$++\}\\
\hspace*{0.3in} {\bf (3.1)} $D^2$-sample a multi-set $S$ of $N$ points &
\hspace*{0.1in} - For all $s \in S$: \\
\hspace*{0.6in} from $X$ with respect to center set $C$ &
\hspace*{0.3in} - If ($\exists j, y$ s.t. $y \in S_j$ \& {\tt SameCluster($s, y$)})\\
\hspace*{0.3in} {\bf (3.2)} $s \leftarrow $ {\tt UncoveredCluster($C, S, R$)} &
\hspace*{0.5in} - $S_j \leftarrow S_j \cup \{s\}$\\
\hspace*{0.3in} {\bf (3.3)} $T \leftarrow $ {\tt UniformSample($X, C, s$)} &
\hspace*{0.3in} - else \\
\hspace*{0.3in} {\bf (3.4)} If ($|T| < M$) continue &
\hspace*{0.5in} - Let $i$ be any index s.t. $S_i$ is empty \\
\hspace*{0.3in} {\bf (3.5)} $R \leftarrow R \cup \{s\}$ &
\hspace*{0.5in} - $S_i \leftarrow \{s\}$\\
\hspace*{0.3in} {\bf (3.6)} $C \leftarrow C \cup \mu(T)$  &
\hspace*{0.1in} - Let $S_i$ be the largest set such that $i > |R|$\\
\hspace*{0.1in} {\bf (4)} return($C$) &
\hspace*{0.1in} - Let $s \in S_i$ be the element with smallest \\ \cline{1-1}
{\tt UniformSample($X, C, s$)} &
 \hspace{0.2in} value of $\Phi(C, \{s\})$ in $S_i$ \\
\hspace*{0.1in} - $T \leftarrow \emptyset$  &
\hspace*{0.1in} - return($s$)\\
\hspace*{0.1in} - For $i$ = $1$ to $L$: & \\
\hspace*{0.3in} - $D^2$-sample a point $x$ from $X$ with respect to center set $C$ & \\
\hspace*{0.3in} - If ({\tt SameCluster($s, x$)}) & \\
\hspace*{0.5in} - With probability $\left(\frac{\veps}{128} \cdot \frac{\Phi(C, \{s\})}{\Phi(C, \{x\})} \right)$ add $x$ in multi-set $T$ & \\
\hspace*{0.1in} - return($T$) & \\
 \hline
\end{tabular}
\caption{Approximation algorithm for $k$-means(top-left frame).
Note that $\mu(T)$ denotes the centroid of $T$ and $D^2$-sampling w.r.t. empty center set $C$ means just uniform sampling. The algorithm {\tt UniformSample($X, C, s$)} (bottom-left) returns a uniform sample of size $\Omega(1/\veps)$ (w.h.p.) from the optimal cluster to which point $s$ belongs.}
\label{fig:algo}
\end{table}

\noindent
We now formally prove the approximation guarantee of the {\tt Query-$k$-means} algorithm.

\begin{theorem}\label{thm:main-query}
Let $0 < \veps \leq 1/2$, $k$ be any positive integer, and $X \subseteq \R^d$. 
There exists an algorithm that runs in time $O(n d k^9/\veps^4)$ and with probability at least $\frac{1}{4}$ outputs a center set $C$ such that $\Phi(C, X) \leq (1 + \veps) \cdot \Delta_k(X)$. Moreover, the number of same-cluster queries used by 
the algorithm is $O(k^9/\veps^4)$.
\end{theorem}

Note that the success probability of the algorithm may be boosted by repeating it a constant number of times. 
This will also prove our main theorem (that is, Theorem~\ref{thm:main1}). We will assume that the dataset $X$ satisfies $(k, \veps)$-irreducibility property defined next. 
We will later drop this assumption using a simple argument and show that the result holds for {\em all} datasets. 
This property was also used in some earlier works~\cite{kss, jks}.

\begin{definition}[Irreducibility]
Let $k$ be a positive integer and $0 < \veps \leq 1$. A given dataset $X \subseteq \R^d$ is said to be $(k, \veps)$-irreducible iff
\[
\Delta_{k-1}(X) \geq (1 + \veps) \cdot \Delta_k(X).
\]
\end{definition}

Qualitatively, what the irreducibility assumption implies is that the optimal solution for the $(k-1)$-means problem does not give a $(1 + \veps)$-approximation to the $k$-means problem.
The following useful lemmas are well known facts.

\begin{lemma}\label{lemma:fact}
For any dataset $X \subseteq \R^d$ and a point $c \in \R^d$, we have:
\[
\Phi(\{c\}, X) = \Phi(\mu(X), X) + |X| \cdot ||c - \mu(X)||^2.
\]
\end{lemma}

\begin{lemma}[Approximate Triangle Inequality] \label{lemma:approx}
For any three points $p,q,r\in \R^d$, we have 
\[ ||p-q||^2\leq 2(||p-r||^2+||r-q||^2)\]
\end{lemma}


Let $X_1, ..., X_k$ be optimal clusters of the dataset $X$ for the $k$-means objective. 
Let $\cs{1}, ..., \cs{k}$ denote the corresponding optimal $k$ centers. 
That is, $\forall i, \cs{i} = \mu(X_i)$.
For all $i$, let $m_i = |X_i|$.
Also, for all $i$, let $r_i = \frac{\sum_{x \in X_i} ||x - \cs{i}||^2}{m_i}$.
The following useful lemma holds due to irreducibility.\footnote{This is Lemma~4 from \cite{jks}. We give the proof here for self-containment.}
\begin{lemma}\label{lemma:query-1}
For all $1 \leq i < j \leq k$, $||\cs{i} - \cs{j}||^2 \geq \veps \cdot (r_i + r_j)$.
\end{lemma}
\begin{proof}
For the sake of contradiction, assume that $||\cs{i} - \cs{j}||^2 < \veps \cdot (r_i + r_j)$.
WLOG assume that $m_i > m_j$.
We have:
\begin{eqnarray*}
\Phi(\{\cs{i}\}, X_i \cup X_j) &=& m_i r_i + m_j r_j + m_j ||\cs{i} - \cs{j}||^2 \quad \textrm{(using Lemma~\ref{lemma:fact})} \\
&\leq& m_i r_i + m_j r_j + m_j \cdot \veps \cdot (r_i + r_j) \\
&\leq& (1+\veps) \cdot m_i r_i + (1+\veps)\cdot m_j r_j \quad \textrm{(since $m_i > m_j$)}\\
&\leq& (1+ \veps) \cdot \Phi(\{\cs{i}, \cs{j}\}, X_i \cup X_j)
\end{eqnarray*}
This implies that the center set $\{\cs{1}, ..., \cs{k}\} \setminus \{\cs{j}\}$ gives a $(1 + \veps)$-approximation to the $k$-means objective which contradicts with the $(k, \veps)$-irreducibility assumption on the data.
\qed
\end{proof}

Consider the algorithm {\tt Query-$k$-means} in Figure~\ref{fig:algo}. 
Let $C_i = \{c_1, ..., c_i\}$ denote the set of centers at the end of the $i^{\textrm{th}}$ iteration of the for loop. That is, $C_i$ is the same as variable $C$ at the end of iteration $i$. We will prove Theorem~\ref{thm:main-query} by inductively arguing that for every $i$, there are $i$ distinct clusters for which centers in $C_i$ are good in some sense. 
Consider the following invariant:
\begin{quote}
$\mathbf{P(i)}$: 
There exists a set of $i$ distinct clusters $X_{j_1}, X_{j_2}, ..., X_{j_i}$ such that 
$$\forall r \in \{1, ..., i\}, \Phi(\{c_r\}, X_{j_r}) \leq \left(1+\frac{\veps}{16} \right) \cdot \Delta_1(X_{j_r}).$$
\end{quote}

Note that a trivial consequence of $P(i)$ is $\Phi(C_i, X_{j_1} \cup ... \cup X_{j_i}) \leq (1+\frac{\veps}{16}) \cdot \sum_{r=1}^{i} \Delta_1(X_{j_r})$.
We will show that for all $i$, $P(i)$ holds with probability at least $(1 - 1/k)^i$.
Note that the theorem follows if $P(k)$ holds with probability at least $(1-1/k)^k$.
We will proceed using induction.
The base case $P(0)$ holds since $C_0$ is the empty set.
For the inductive step, assuming that $P(i)$ holds with probability at least $(1-1/k)^i$ for some arbitrary $i \geq 0$, we will show that $P(i+1)$ holds with probability at least $(1-1/k)^{i+1}$.
We condition on the event $P(i)$ (that is true with probability at least $(1 - 1/k)^{i}$).
Let $C_i$ and $X_{j_1}, ..., X_{j_i}$ be as guaranteed by the invariant $P(i)$.
For ease of notation and without loss of generality, let us assume that the index $j_r$ is $r$.
So, $C_i$ gives a good approximation w.r.t. points in the set $X_1 \cup X_2 \cup .... \cup X_i$ and these clusters may be thought of as ``covered" clusters (in the approximation sense).
Suppose we $D^2$-sample a point $p$ w.r.t. center set $C_i$. The probability that $p$ belongs to some ``uncovered cluster" $X_r$ where $r \in [i+1,k]$ is given as $\frac{\Phi(C_i, X_r)}{\Phi(C_i, X)}$. If this quantity is small, then the points sampled using $D^2$ sampling in subsequent iterations may not be good representatives for the uncovered clusters. This may cause the analysis to break down. However, we argue that since our data is $(k, \veps)$-irreducible, this does not occur.
\footnote{This is Lemma 5 in \cite{jks}. We give the proof for self-containment.}
\begin{lemma}
$\frac{\Phi(C_i, X_{i+1} \cup ... \cup X_{k})}{\Phi(C_i, X)} \geq \frac{\veps}{4}$.
\end{lemma}
\begin{proof}
For the sake of contradiction, assume that the above statement does not hold. Then we have:
\begin{eqnarray*}
\Phi(C_i, X) &=& \Phi(C_i, X_1\cup ...\cup X_i) + \Phi(C_i,X_{i+1} \cup ... \cup X_k) \\
&\leq& \Phi(C_i, X_1\cup ...\cup X_i) + \frac{(\veps/4)}{1 - (\veps/4)} \cdot \Phi(C_i, X_1\cup ...\cup X_i) \quad \textrm{(using our assumption)} \\
&=& \frac{1}{1 - \veps/4} \cdot \Phi(C_i, X_1\cup ...\cup X_i) \\
&\leq& \frac{1 + \veps/16}{1 - \veps/4} \cdot \sum_{j=1}^{i} \Delta_1(X_j) \quad \textrm{(using invariant)} \\
&\leq& (1 + \veps) \cdot \sum_{j=1}^{k} \Delta_1(X_j)
\end{eqnarray*}
This contradicts with the $(k, \veps)$-irreducibility of $X$.\qed
\end{proof}

The following simple corollary of the above lemma will be used in the analysis later.
\begin{corollary}\label{cor:1}
There exists an index $j \in \{i+1, ..., k\}$ such that $\frac{\Phi(C_i, X_j)}{\Phi(C_i, X)} \geq \frac{\veps}{4k}$.
\end{corollary}

The above corollary says that there is an uncovered cluster which will have a non-negligible representation in the set $S$ that is sampled in iteration $(i+1)$ of the algorithm {\tt Query-$k$-means}.
The next lemma shows that conditioned on sampling from an uncovered cluster $l \in \{i+1, ..., k\}$, the probability of sampling a point $x$ from $X_l$ is at least $\frac{\veps}{64}$ times its sampling probability if it were sampled uniformly from $X_l$ (i.e., with probability at least $\frac{\veps}{64} \cdot \frac{1}{m_l}$).
\footnote{This is Lemma~6 from \cite{jks}. We give the proof for self-containment.}

\begin{lemma}\label{lemma:cond-prob}
For any $l \in \{i+1, ..., k\}$ and $x \in X_{l}$,  $\frac{\Phi(C_i, \{x\})}{\Phi(C_i, X_l)} \geq \frac{\veps}{64} \cdot \frac{1}{m_l}$.
\end{lemma}
\begin{proof}
Let $t \in \{1, ..., i\}$ be the index such that $x$ is closest to $c_t$ among all centers in $C_i$. We have:
\begin{eqnarray*}
\Phi(C_i, X_l) &=& m_l \cdot r_l + m_l \cdot ||\cs{l} - c_t||^2 \quad \textrm{(using Lemma~\ref{lemma:fact})}\\
&\leq& m_l \cdot r_l + 2 m_l \cdot (||\cs{l} - \cs{t}||^2 + ||\cs{t} - c_t||^2) \quad \textrm{(using Lemma~\ref{lemma:approx})}\\
&\leq& m_l \cdot r_l + 2 m_l \cdot (||\cs{l} - \cs{t}||^2 + \frac{\veps}{16} \cdot r_t) \quad \textrm{(using invariant and Lemma~\ref{lemma:fact})}
\end{eqnarray*}
Also, we have:
\begin{eqnarray*}
\Phi(C_i, \{x\}) = ||x - c_t||^2 &\geq& \frac{||x - \cs{t}||^2}{2} - ||\cs{t} - c_t||^2 \quad \textrm{(using Lemma~\ref{lemma:approx})}\\
&\geq& \frac{||\cs{l} - \cs{t}||^2}{8} - ||\cs{t} - c_t||^2 \quad \textrm{(since $||x - \cs{t}|| \geq ||\cs{l} - \cs{t}||/2$)}\\
&\geq&  \frac{||\cs{l} - \cs{t}||^2}{8} - \frac{\veps}{16} \cdot r_t \quad \textrm{(using invariant and Lemma~\ref{lemma:fact})}\\
&\geq& \frac{||\cs{l} - \cs{t}||^2}{16} \quad \textrm{(using Lemma~\ref{lemma:query-1})}
\end{eqnarray*}
Combining the inequalities obtained above, we get the following:
\begin{eqnarray*}
\frac{\Phi(C_i, \{x\})}{\Phi(C_i, X_l)} &\geq& \frac{||\cs{l} - \cs{t}||^2}{16 \cdot m_l \cdot \left( r_l + 2 ||\cs{l} - \cs{t}||^2 + \frac{\veps}{8} \cdot r_t\right)} \\
&\geq& \frac{1}{16 \cdot m_l} \cdot \frac{1}{(1/\veps) + 2 + (1/8)} \geq \frac{\veps}{64} \cdot \frac{1}{m_l}
\end{eqnarray*}
This completes the proof of the lemma.\qed
\end{proof}

With the above lemmas in place, let us now get back to the inductive step of the proof. Let $J \subseteq \{i+1, ..., k\}$ denote the subset of indices (from the uncovered cluster indices) such that $\forall j \in J, \frac{\Phi(C_i, X_j)}{\Phi(C_i, X)} \geq \frac{\veps}{8k}$. For any index $j \in J$, let $Y_j \subseteq X_j$ denote the subset of points in $X_j$ such that $\forall y \in Y_j, \frac{\Phi(C_i, \{y\})}{\Phi(C_i, X_j)} \leq \frac{2}{m_j}$. That is, $Y_j$ consists of all the points such that the conditional probability of sampling any point $y$ in $Y_j$, given that a point is sampled from $X_j$, is upper bounded by $2/m_j$. Note that from Lemma~\ref{lemma:cond-prob}, the conditional probability of sampling a point $x$ from $X_j$, given that a point is sampled from $X_j$, is lower bounded by $\frac{\veps}{64} \cdot \frac{1}{m_j}$. This gives the following simple and useful lemma.
\begin{lemma}\label{lemma:9}
For all $j \in \{i+1, ..., k\}$ the following two inequalities hold:
\begin{enumerate}
\item $\frac{\Phi(C_i, Y_j)}{\Phi(C_i, X)} \geq \frac{\veps}{128} \cdot \frac{\Phi(C_i, X_j)}{\Phi(C_i, X)}$, and
\item For any $y \in Y_j$ and any $x \in X_j$, $\frac{\veps}{128} \cdot \Phi(C_i, \{y\}) \leq \Phi(C_i, \{x\})$.
\end{enumerate}
\end{lemma}
\begin{proof}
Inequality (1) follows from the fact that $|Y_j| \geq m_j/2$, and $\frac{\Phi(C_i, \{y\})}{\Phi(C_i, X_j)} \geq \frac{\veps}{64} \cdot \frac{1}{m_j}$ for all $y \in X_j$. Inequality (2) follows from the fact that for all $x \in X_j, \frac{\Phi(C_i, \{x\})}{\Phi(C_i, X_j)} \geq \frac{\veps}{64} \cdot \frac{1}{m_j}$ and for all $y \in Y_j, \frac{\Phi(C_i, \{y\})}{\Phi(C_i, X_j)} \leq \frac{2}{m_j}$.\qed
\end{proof}

Let us see the outline of our plan before continuing with our formal analysis. What we hope to get in line (3.2) of the algorithm is a point $s$ that belongs to one of the uncovered clusters with index in the set $J$. That is, $s$ belongs to an uncovered cluster that is likely to have a good representation in the $D^2$-sampled set $S$ obtained in line (3.1). In addition to $s$ belonging to $X_j$ for some $j \in J$, we would like $s$ to belong to $Y_j$. This is crucial for the uniform sampling in line (3.3) to succeed. We will now show that the probability of $s$ returned in line (3.2) satisfies the above conditions is large.

\begin{lemma}
Let $S$ denote the $D^2$-sample obtained w.r.t. center set $C_i$ in line (3.1) of the algorithm.
\[
\pr[\exists j \in J \textrm{ such that $S$ does not contain any point from $Y_j$}] \leq \frac{1}{4k}.
\]
\end{lemma}
\begin{proof}
We will first get bound on the probability for a fixed $j \in J$ and then use the union bound. From property (1) of Lemma~\ref{lemma:9}, we have that for any $j \in J$, $\frac{\Phi(C_i, Y_j)}{\Phi(C_i, X)} \geq \frac{\veps}{128} \cdot \frac{\veps}{8k} = \frac{\veps^2}{(2^{10})k}$. Since the number of sampled points is $N = \frac{(2^{12})k^3}{\veps^2}$, we get that the probability that $S$ has no points from $Y_j$ is at most $\frac{1}{4k^2}$. Finally, using the union bound, we get the statement of the lemma.\qed
\end{proof}

\begin{lemma}
Let $S$ denote the $D^2$-sample obtained w.r.t. center set $C_i$ in line (3.1) of the algorithm and let $S_j$ denote the representatives of $X_j$ in $S$. Let $max = \arg\max_{j \in \{i+1, ..., k\}} |S_j|$. Then $\pr[max \notin J] \leq \frac{1}{4k}$.
\end{lemma}
\begin{proof}
From Corollary~\ref{cor:1}, we know that there is an index $j \in \{i+1, ..., k\}$ such that $\frac{\Phi(C_i, X_j)}{\Phi(C_i, X)} \geq \frac{\veps}{4k}$. 
Let $\alpha = N \cdot \frac{\veps}{4k}$. 
The expected number of representatives from $X_j$ in $S$ is at least $\alpha$. 
So, from Chernoff bounds,
we have:
\[
\pr[|S_j| \leq 3\alpha/4] \leq e^{-\alpha/32}
\]
On the other hand, for any $r \in \{i+1, ..., k\} \setminus J$, the expected number of points in $S$ from $X_r$ is at most $\frac{\veps}{8k} \cdot N = \alpha/2$. So, from Chernoff bounds, 
we have:
\[
\pr[|S_r| > 3\alpha/4] \leq e^{-\alpha/24}
\]
So, the probability that there exists such an $r$ is at most $k \cdot e^{-\alpha/24}$ by union bound.
Finally, the probability that $max \notin J$ is at most $\pr[|S_j| \leq 3\alpha/4] + \pr[\exists r \in \{i+1, ..., k\}\setminus J | |S_r|> 3\alpha/4]$ which is at most $\frac{1}{4k}$ due to our choice of $N = \frac{(2^{12})k^3}{\veps^2}$.\qed
\end{proof}

From the previous two lemmas, we get that with probability at least $(1 - \frac{1}{2k})$, the $s$ returned in line (3.2) belongs to $Y_j$ for some $j \in J$.
Finally, we will need the following claim to argue that the set $T$ returned in line (3.3) is a uniform sample from one of the sets $X_j$ for $j \in \{i+1, ..., k\}$.

\begin{lemma}
Let $S$ denote the $D^2$-sample obtained w.r.t. center set $C_i$ in line (3.1) and $s$ be the point returned in line (3.2) of the algorithm. Let $j$ denote the index of the cluster to which $s$ belongs. If $j \in J$ and $s \in Y_j$, then with probability at least $(1-\frac{1}{4k})$, $T$ returned in line (3.3) is a uniform sample from $X_j$ with size at least $\frac{64k}{\veps}$.
\end{lemma}
\begin{proof}
Consider the call to sub-routine {\tt UniformSample($X, C_i, s$)} with $s$ as given in the statement of the lemma. If $j$ is the index of the cluster to which $s$ belongs, then $j \in J$ and $s \in Y_j$. Let us define $L$ random variables $Z_1, ..., Z_L$ one for every iteration of the sub-routine. These random variables are defined as follows: for any $r \in [1,L]$, if the sampled point $x$ belongs to the same cluster as $s$ and it is picked to be included in multi-set $S$, then $Z_r  = x$, otherwise $Z_r = \bot$. We first note that for any $r$ and any $x, y \in X_j$, $\pr[Z_r = x] = \pr[Z_r = y]$.
This is because for any $x \in X_j$, we have $\pr[Z_r = x] = \frac{\Phi(C_i, \{x\})}{\Phi(C_i, X)} \cdot \frac{\frac{\veps}{128}  \cdot \Phi(C_i, \{s\})}{\Phi(C_i, \{x\})} = \frac{\veps}{128} \cdot \frac{\Phi(C_i, \{s\})}{\Phi(C_i, X)}$. It is important to note that $\frac{\frac{\veps}{128}  \cdot \Phi(C_i, \{s\})}{\Phi(C_i, \{x\})} \leq 1$ from property (2) of Lemma~\ref{lemma:9} and hence valid in the probability calculations above.

Let us now obtain a bound on the size of $T$. Let $T_r = I(Z_r)$ be the indicator variable that is $1$ if $Z_r \neq \bot$ and $0$ otherwise. Using the fact that $j \in J$, we get that for any $r$:
\[
\E[T_r] = \pr[T_r = 1] = \frac{\veps}{128} \cdot \frac{\sum_{x \in X_j} \Phi(C_i, \{s\})}{\Phi(C_i, X)} \geq \frac{\veps}{128} \cdot \frac{\veps}{8k} \cdot \frac{\veps}{64} = \frac{\veps^3}{(2^{16})k}.
\]
Given that $L = \frac{2^{23}k^2}{\veps^4}$, applying Chernoff bounds, 
we get the following:
\[
\pr \left[|T| \geq \frac{64k}{\veps}\right] = 1-\pr \left[|T| \leq \frac{64k}{\veps} \right] \geq \left(1 - \frac{1}{4k} \right)
\]
This completes the proof of the lemma.\qed
\end{proof}

Since a suitable $s$ (as required by the above lemma) is obtained in line (3.2) with probability at least $(1-\frac{1}{2k})$, the probability that $T$ obtained in line (3.3) is a uniform sample from some uncovered cluster $X_j$ is at least $(1-\frac{1}{2k}) \cdot (1 - \frac{1}{4k})$.
Finally, the probability that the centroid $\mu(T)$ of the multi-set $T$ that is obtained is a good center for $X_j$ is at least $(1 - \frac{1}{4k})$ using Inaba's lemma.
Combining everything, we get that with probability at least $(1 - \frac{1}{k})$ an uncovered cluster will be covered in the $i^{th}$ iteration.
This completes the inductive step and hence the approximation guarantee of $(1+\veps)$ holds for any dataset that satisfies the $(k, \veps)$-irreducibility assumption.
For the number of queries and running time, note that every time sub-routine {\tt UncoveredCluster} is called, it uses at most $kN$ same cluster queries. For the subroutine {\tt UniformSample}, the number of same-cluster queries made is $L$. So, the total number of queries is $O(k (k N + L)) = O(k^5/\veps^4)$.
More specifically, we have proved the following theorem.

\begin{theorem}\label{thm:main-query-irreducible}
Let $0 < \veps \leq 1/2$, $k$ be any positive integer, and $X \subseteq \R^d$ such that $X$ is $(k, \veps)$-irrducible.
Then {\tt Query-$k$-means($X, k, \veps$)} runs in time $O(nd k^5/\veps^4)$ and with probability at least $(1/4)$ outputs a center set $C$ such that $\Phi(C, X) \leq (1 + \veps) \cdot \Delta_k(X)$.
Moreover, the number of same-cluster queries used by {\tt Query-$k$-means($X, k, \veps$)} is $O(k^5/\veps^4)$.
\end{theorem}

To complete the proof of Theorem~\ref{thm:main-query}, we must remove the irreducibility assumption in the above theorem.
We do this by considering the following two cases:
\begin{enumerate}
\item Dataset $X$ is $(k, \frac{\veps}{(4 + \veps/2)k})$-irreducible.

\item Dataset $X$ is not $(k, \frac{\veps}{(4 + \veps/2)k})$-irreducible.
\end{enumerate}

In the former case, we can apply Theorem~\ref{thm:main-query-irreducible} to obtain Theorem~\ref{thm:main-query}.
Now, consider the latter case.
Let $1 < i \leq k$ denote the largest index such that $X$ is $(i, \frac{\veps}{(1+\veps/2)k})$-irreducible, otherwise $i = 1$. Then we have:
\[
\Delta_i(X) \leq \left( 1 + \frac{\veps}{(4+\veps/2)k}\right)^{k-i} \cdot \Delta_k(X) \leq \left(1 + \frac{\veps}{4} \right) \cdot \Delta_k(X).
\]
This means that a $(1 + \veps/4)$-approximation for the $i$-means problem on the dataset $X$ gives the desired approximation for the $k$-means problem.
Note that our approximation analysis works for the $i$-means problem with respect to the algorithm being run only for $i$ steps in line (3) (instead of $k$).
That is, the centers sampled in the first $i$ iterations of the algorithm give a $(1 + \veps/16)$-approximation for the $i$-means problem for any fixed $i$.
This simple observation is sufficient for Theorem~\ref{thm:main-query}.
Note since Theorem~\ref{thm:main-query-irreducible} is used with value of the error parameter as $O(\veps/k)$, the bounds on the query and running time get multiplied by a factor of $k^4$.

\section{Query Lower Bound (proof of Theorem~\ref{thm:main2})}\label{sec:main2}
In this section, we will obtain a conditional lower bound on the number of same-cluster queries assuming the Exponential Time Hypothesis (ETH). This hypothesis has been used to obtain lower bounds in various different contexts (see \cite{m16} for reference). We start by stating the Exponential Time Hypothesis (ETH).

\begin{quote}
\underline{\bf Hypothesis 1} {\it (Exponential Time Hypothesis (ETH)\cite{IP01,IPZ01})}: There does not exist an algorithm that can decide whether any $3$-SAT formula with $m$ clauses is satisfiable with running time $2^{o(m)}$.
\end{quote}

Since we would like to obtain lower bounds in the approximation domain, we will need a {\em gap} version of the above ETH hypothesis. The following version of the PCP theorem will be very useful in obtaining a gap version of ETH.

\begin{theorem}[Dinur's PCP Theorem~\cite{Din07}]
For some constants $\veps, d > 0$, there exists a polynomial time reduction that takes a $3$-SAT formula $\psi$ with $m$ clauses and produces another $3$-SAT formula $\phi$ with $m' = O(m\ polylog\ m)$ clauses such that:
\begin{itemize}
\item If $\psi$ is satisfiable, then $\phi$ is satisfiable,
\item if $\psi$ is unsatisfiable, then $val(\phi) \leq 1 - \veps$, and
\item each variable of $\phi$ appears in at most $d$ clauses.
\end{itemize}
Here $val(\phi)$ denotes the maximum fraction of clauses of $\phi$ that are satisfiable by any assignment.
\end{theorem}

The following new hypothesis follows from ETH and will be useful in our analysis.

\begin{quote}
\underline{\bf Hypothesis 2}: There exists constants $\veps, d > 0$ such that the following holds: There does not exist an algorithm that, given a $3$-SAT formula $\psi$ with $m$ clauses with each variable appearing in at most $d$ clauses, distinguishes whether $\psi$ is satisfiable or $val(\psi) \leq (1 - \veps)$, runs in time $2^{\Omega \left(\frac{m}{poly \log m} \right)}$.
\end{quote}

The following simple lemma follows from Dinur's PCP theorem given above.

\begin{lemma}\label{lemma:H1}
If Hypothesis 1 holds, then so does Hypothesis 2.
\end{lemma}

We now see a reduction from the gap version of $3$-SAT to the gap version of the Vertex Cover problem that will be used to argue the hardness of the $k$-means problem. The next result is a standard reduction and can be found in a survey by Luca Trevisan~\cite{luca-survey}.

\begin{lemma}\label{lemma:sat-vc}
Let $\veps, d > 0$ be some constants.
There is a polynomial time computable function mapping $3$-SAT instances $\psi$ with $m$ variables and where each variable appears in at most $d$ clauses, into graphs $G_{\psi}$ with $3m$ vertices and maximum degree $O(d)$ such that if $\psi$ is satisfiable, then $G_{\psi}$ has a vertex cover of size at most $2m$ and if $val(\psi) \leq (1 - \veps)$, then every vertex cover of $G_{\psi}$ has size at least $2m(1+\veps/2)$.
\end{lemma}

We formulate the following new hypothesis that holds given that hypothesis 2 holds. Eventually, we will chain all these hypothesis together.

\begin{quote}
{\bf Hypothesis 3}: There exists constants $\veps, d > 0$ such that the following holds: There does not exist an algorithm that, given a graph $G$ with $n$ vertices and maximum degree $d$, distinguishes between the case when $G$ has a vertex cover of size at most $2n/3$ and the case when $G$ has a vertex cover of size at least $\frac{2n}{3} \cdot (1 + \veps)$, runs in time $2^{\Omega \left(\frac{n}{poly \log n}\right)}$.
\end{quote}

The following lemma is a simple implication of Lemma~\ref{lemma:sat-vc}

\begin{lemma}\label{lemma:H2}
If Hypothesis 2 holds, then so does Hypothesis 3.
\end{lemma}

We are getting closer to the $k$-means problem that has a reduction from the vertex cover problem on triangle-free graphs~\cite{acks15}.
So, we will need reductions from vertex cover problem to vertex cover problem on triangle-free graphs and then to the $k$-means problem.
These two reductions are given by Awasthi \etal \cite{acks15}.

\begin{lemma}[Follows from Theorem 21~\cite{acks15}]\label{lemma:vc-tvc}
Let $\veps, d > 0$ be some constants. There is a polynomial-time computable function mapping any graph $G = (V, E)$ with maximum degree $d$ to a triangle-free graph $\hat{G} = (\hat{V}, \hat{E})$ such that the following holds:
\begin{itemize}
\item $|\hat{V}| = poly(d, 1/\veps) \cdot |V|$ and maximum degree of vertices in $\hat{G}$ is $O(d^3/\veps^2)$, and
\item $\left( 1 - \frac{|VC(G)|}{|V|} \right) \leq \left( 1 - \frac{|VC(\hat{G})|}{|\hat{V}|} \right) \leq (1 + \veps) \cdot \left( 1 - \frac{|VC(G)|}{|V|}\right)$.
\end{itemize}
Here $VC(G)$ denote the size of the minimum vertex cover of $G$.
\end{lemma}

We can formulate the following hypothesis that will follow from Hypothesis 3 using the above lemma.

\begin{quote}
{\bf Hypothesis 4}: There exists constants $\veps, d > 0$ such that the following holds: There does not exist an algorithm that, given a triangle-free graph $G$ with $n$ vertices and maximum degree $d$, distinguishes between the case when $G$ has a vertex cover of size at most $\frac{2n}{3}$ and the case when $G$ has a vertex cover of size at least $\frac{2n}{3} \cdot (1 + \veps)$, runs in time $2^{\Omega \left(\frac{n}{poly \log n}\right)}$.
\end{quote}

The next claim is a simple application of Lemma~\ref{lemma:vc-tvc}.

\begin{lemma}\label{lemma:H3}
If Hypothesis 3 holds, then so does Hypothesis 4.
\end{lemma}

Finally, we use the reduction from the vertex cover problem in triangle-free graphs to the $k$-means problem to obtain the hardness result for the $k$-means problem.
We will use the following reduction from Awasthi \etal \cite{acks15}.

\begin{lemma}[Theorem 3~\cite{acks15}]
There is an efficient reduction from instances of Vertex Cover (in triangle free graphs) to those of $k$-means that satisfies the following properties:
\begin{itemize}
\item if the Vertex Cover instance has value $k$, then the $k$-means instance has cost at most $(m-k)$
\item if the Vertex Cover instance has value at least $k(1+\veps)$, then the optimal $k$-means cost is at least $m - (1 - \Omega(\veps))k$. Here $\veps$ is some fixed constant $> 0$.
\end{itemize}
Here $m$ denotes the number of edges in the vertex cover instance.
\end{lemma}

The next hypothesis follows from Hypothesis 4 due to the above lemma.

\begin{quote}
{\bf Hypothesis 5}: There exists constant $c > 1$ such that the following holds: There does not exist an algorithm that gives an approximation guarantee of $c$ for the $k$-means problem that runs in time $poly(n, d) \cdot 2^{\Omega\left( \frac{k}{poly \log k}\right)}$.
\end{quote}

\begin{claim}\label{lemma:H4}
If Hypothesis 4 holds, then so does Hypothesis 5.
\end{claim}

Now using Lemmas~\ref{lemma:H1}, \ref{lemma:H2}, \ref{lemma:H3}, and \ref{lemma:H4}, get the following result.

\begin{lemma}
If the Exponential Time Hypothesis (ETH) holds then there exists a constant $c > 1$ such that any $c$-approximation algorithm for the $k$-means problem cannot have running time better than $poly(n, d) \cdot 2^{\Omega \left( \frac{k}{poly \log k}\right)}$.
\end{lemma}

This proves Theorem~\ref{thm:main2} since if there is a query algorithm that runs in time $poly(n, d, k)$ and makes $\frac{k}{poly \log k}$ same-cluster queries, then we can convert it to a non-query algorithm that runs in time $poly(n, d) \cdot 2^{\frac{k}{poly \log k}}$ in a brute-force manner by trying out all possible answers for the queries and then picking the best $k$-means solution.

\section{Query Approximation Algorithm with Faulty Oracle}\label{sec:faulty}

In this section, we describe how to extend our approximation algorithm for $k$-means clustering in the SSAC framework when the oracle is {\em faulty}. 
That is, the answer to the same-cluster query may be incorrect.
Let us denote the faulty same-cluster oracle as $\OA^E$. 
We consider the following error model: for a query with points $u$ and $v$, the query answer $\OA^E(u,v)$ is wrong independently with probability at most $q$ that is strictly less than $1/2$. 
More specifically, if $u$ and $v$ belong to the same optimal cluster, then $\OA^E(u,v)=1$ with probability at least $(1-q)$ and $\OA^E(u,v)=0$ with probability at most $q$. 
Similarly, if $u$ and $v$ belong to different optimal clusters, then $\OA^E(u,v)=1$ with probability at most $q$ and $\OA^E(u,v)=0$ with probability at least $(1-q)$.

The modified algorithm giving $(1+\eps)$-approximation for $k$-means with faulty oracle $\OA^E$ is given in Figure \ref{fig:algo3}. 
Let $X_1,\ldots,X_k$ denote the $k$ optimal clusters for the dataset $X$. 
Let $C=\{c_1,\ldots,c_i\}$ denote the set of $i$ centers chosen by the algorithm at the end of iteration $i$. 
Let $S$ denote the sample obtained using $D^2$-sampling in the $(i+1)^{\text{st}}$ iteration. 
The key idea for an efficient $(1+\veps)$-approximation algorithm for $k$-means in the SSAC framework with a {\em perfect} oracle was the following. Given a sample $S$, we could compute using at most $k|S|$ same-cluster queries the partition $S_1,\ldots,S_k$ of $S$ among the $k$ optimal clusters such that $S_j=S\cap X_j$ for all $j$. 
In the following, we discuss how we achieve this partitioning of $S$ among $k$ optimal clusters when the oracle $\OA^E$ is faulty.

We reduce the problem of finding the partitions of $S$ among the optimal clusters to the problem of recovering dense (graph) clusters in a {\em stochastic block model} (SBM). 
An instance of an SBM is created as follows. Given any arbitrary partition $V_1,\ldots,V_k$ of a set $V$ of vertices, an edge is added between two vertices belonging to the same partition with probability at least $(1-q)$ and between two vertices in different partitions with probability at most $q$. 
We first construct an instance $I$ of an SBM using the sample $S$. 
By querying the oracle $\OA^E$ with all pairs of points $u,v$ in $S$, we obtain a graph $I$ on $|S|$ vertices, where vertices in $I$ correspond to the points in $S$, and an edge exists in $I$ between vertices $u$ and $v$ if $\OA^E(u,v)=1$. 
Since oracle $\OA^E$ errs with probability at most $q$, for any $u,v\in S_j$ for some $j\in [k]$, there is an edge between $u$ and $v$ with probability at least  $(1-q)$. 
Similarly, there is an edge $(u,v)\in I$ for any two points $u\in S_y$ and $v\in S_z, y\neq z$ belonging to different optimal clusters with probability at most  $q$. 
Note that the instance $I$ created in this manner would be an instance of an SBM. 
Since $q<1/2$, this procedure, with high probability, creates more edges between vertices belonging to the same partition than the number of edges between vertices in different partitions. 
Intuitively, the partitions of $S$ would correspond to the dense (graph) clusters in SBM instance $I$, and if there were no errors, then each partition would correspond to a clique in $I$. 
One way to figure out the partitions $S_1,\ldots,S_k$ would be to retrieve the dense (graph) clusters from the instance $I$. 
Ailon et al.~\cite{ACX2015} gave a randomized algorithm to retrieve all large clusters of any SBM instance. 
Their algorithm on a graph of $n$ vertices retrieves all clusters of size at least $\sqrt{n}$ with high probability. 
Their main result in our context is given as follows.

\begin{table}[htbp]
\centering
\begin{tabular}{| l | l |}
\hline
{\bf Constants}: $N = \frac{(2^{13})k^3}{\eps^2}$, $M = \frac{64k}{\veps}$, $L = \frac{(2^{23})k^2}{\veps^4}$ & \\
\cline{1-1}
{\tt Faulty-Query-$k$-means($X, k, \veps$)} &
{\tt UncoveredCluster($C, S, J$)}\\
\hspace*{0.1in} {\bf (1)} $J \leftarrow \emptyset$ &
\hspace*{0.1in} - For all $i \in \{1, ..., k\}$: $S_i \leftarrow \emptyset$\\
\hspace*{0.1in} {\bf (2)} $C \leftarrow \emptyset$ &
\hspace*{0.1in} - $i \leftarrow 1$ \\
\hspace*{0.1in} {\bf (3)} for $i$ = $1$ to $k$  &
\hspace*{0.1in} - For all $y \in J$: \{$S_i \leftarrow y$; $i$++\}\\
\hspace*{0.3in} {\bf (3.1)} $D^2$-sample a multi-set $S$ of $N$ points &
\hspace*{0.1in} - $T_1,\ldots,T_l$= {\tt PartitionSample($S$)} for $l<k$\\
\hspace*{0.6in} from $X$ with respect to center set $C$ &
\hspace*{0.1in} - for $j=1,\ldots,l$\\
\hspace*{0.3in} {\bf (3.2)} $s \leftarrow $ {\tt UncoveredCluster($C, S, J$)} &
\hspace*{0.3in} - if IsCovered$(C,T_j)$ is FALSE\\
\hspace*{0.3in} {\bf (3.3)} $T \leftarrow $ {\tt UniformSample($X, C, s$)} &
\hspace*{0.5in} - if $\exists t$ such that $S_t=\emptyset$, then $S_t=T_j$ \\
\hspace*{0.3in} {\bf (3.4)} If ($|T| < M$) continue &
\hspace*{0.1in} - Let $S_i$ be the largest set such that $i > |J|$ \\
\hspace*{0.3in} {\bf (3.5)} $J \leftarrow J \cup \{s\}$ &
\hspace*{0.1in} - Let $s \in S_i$ be the element with smallest\\
\hspace*{0.3in} {\bf (3.6)} $C \leftarrow C \cup \mu(T)$  &
\hspace*{0.2in} value of $\Phi(C, \{s\})$ in $S_i$\\
\hspace*{0.1in} {\bf (4)} return($C$) &
\hspace*{0.1in} - return($s$)  \\ \cline{1-1}
{\tt UniformSample($X, C, s$)} &
 \hspace{0.2in} \\ \cline{2-2}
\hspace*{0.1in} - $S \leftarrow \emptyset$  &
\hspace*{0.1in} \\
\hspace*{0.1in} - For $i$ = $1$ to $L$: & {\tt PartitionSample($S$)} \\
\hspace*{0.3in} - $D^2$-sample point $x\in X$ with respect to center set $C$ & - Construct SBM instance $I$ by querying $\OA^E(s,t)$ $\forall s,t\in S$\\
\hspace*{0.3in} - $U=U\cup \{x\}$ & - Run cluster recovery algorithm of Ailon et al.~\cite{ACX2015} on $I$\\
\hspace*{0.1in} - $T_1,\ldots,T_l$ = {\tt PartitionSample($U$)} for $l<k$ & - Return $T_1,\ldots,T_l$ for $l< k$\\ \cline{2-2}
\hspace*{0.1in} - for $j=1,\ldots,l$ & \\
\hspace*{0.3in} - If ({\tt IsCovered($s, T_j$)} is TRUE) & \\
\hspace*{0.5in} - $\forall x\in T_j$, with probability $\left(\frac{\veps}{128} \cdot \frac{\Phi(C, \{s\})}{\Phi(C, \{x\})} \right)$ add $x$ & {\tt IsCovered(C,U)}\\
\hspace*{0.5in} in multi-set $S$ & - for $c \in C$\\
\hspace*{0.1in} - return $(S)$ & \hspace*{0.1in}-- if for majority of $u\in U$, $\OA^E$$(c,u)=1$, Return TRUE \\
\hspace*{0.1in}  & - Return FALSE \\
 \hline
\end{tabular}
\caption{Approximation algorithm for $k$-means (top-left frame) using faulty oracle. Note that $\mu(T)$ denotes the centroid of $T$ and $D^2$-sampling w.r.t. empty center set $C$ means just uniform sampling. The algorithm {\tt UniformSample($X, C, s$)} (bottom-left) returns a uniform sample of size $\Omega(1/\veps)$ (w.h.p.) from the optimal cluster in which point $s$ belongs.}
\label{fig:algo3}
\end{table}

\begin{lemma}[\cite{ACX2015}] \label{recovery1} 
There exists a polynomial time algorithm that, given an instance of a stochastic block model on $n$ vertices, retrieves all clusters of size at least $\Omega(\sqrt{n})$ with high probability, provided $q<1/2$. 
\end{lemma}

We use Lemma \ref{recovery1} to retrieve the large clusters from our SBM instance $I$. We also need to make sure that the sample $S$ is such that its overlap with at least one uncovered optimal cluster is large, where an optimal cluster $S_j$ for some $j$ is {\em uncovered} if $C\cap S_j=\emptyset$. 
More formally, we would require the following: $\exists j\in [k]$ such that $|S_j|\geq \Omega(\sqrt{|S|})$, and $X_j$ is uncovered by $C$. From Corollary \ref{cor:1}, given a set of centers $C$ with $|C|<k$, there exists an uncovered cluster such that any point sampled using $D^2$-sampling would belong to that uncovered cluster with probability at least $\frac{\eps}{4k}$. Therefore, in expectation, $D^2$-sample $S$ would contain at least $\frac{\eps}{4k}|S|$ points from one such uncovered optimal cluster. In order to ensure that this quantity is at least as large as $\sqrt{|S|}$, we need $|S|=\Omega(\frac{16k^2}{\eps^2})$. Our bounds for $N$ and $L$, in the algorithm, for the size of $D^2$-sample $S$ satisfy this requirement with high probability. This follows from a simple application of Chernoff bounds.

\begin{lemma}\label{rec2} For $D^2$-sample $S$ of size at least $\frac{2^{12} k^2}{\eps^2}$, there is at least one partition $S_j=S\cap X_j$ among the partitions returned by the sub-routine {\tt PartitionSample} corresponding to an uncovered cluster $X_j$ with probability at least $(1-\frac{1}{16k})$.\end{lemma}
\begin{proof} From Corollary \ref{cor:1}, for any point $p$ sampled using $D^2$-sampling, the probability that point $p$ belongs to some uncovered cluster $X_j$ is at least $\frac{\eps}{4k}$. In expectation, the number of points sampled from uncovered cluster $X_j$ is $\E[|S_j|]=\frac{\eps|S|}{4k}=\frac{2^{10}k}{\eps}$. Exact recovery using Lemma \ref{recovery1} requires $|S_j|$ to be at least $\frac{2^6 k}{\eps}$. Using Chernoff bounds, the probability of this event is at least $(1-\frac{1}{16k})$. \qed \end{proof}

Following Lemma \ref{rec2}, we condition on the event that there is at least one partition corresponding to an uncovered cluster among the partitions returned by the sub-routine {\tt PartitionSample}. Next, we figure out using the sub-routine {\tt IsCovered} which of the partitions returned by {\tt PartitionSample} are covered and which are uncovered. Let $T_1,\ldots,T_l$ be the partitions returned by {\tt PartitionSample} where $l<k$. Sub-routine {\tt IsCovered} determines whether a cluster is covered or uncovered in the following manner. For each $j\in [l]$, we check whether $T_j$ is covered by some $c\in C$. We query oracle $\OA^E$ with pairs $(v,c)$ for $v\in T_j$ and $c\in C$. If majority of the query answers for some $c\in C$ is $1$, we say cluster $T_j$ is covered by $C$. If for all $c\in C$ and some $T_j$, the majority of the query answers is $0$, then we say $T_j$ is uncovered by $C$. Using Chernoff bounds, we show that with high probability uncovered clusters would be detected. 

\begin{lemma}\label{lemma-isCov} With probability at least $(1-\frac{1}{16k})$, all covered and uncovered clusters are detected correctly by the sub-routine {\tt IsCovered}. \end{lemma}
\begin{proof} First, we figure out the probability that any partition $T_j$ for $j\in [l]$ is detected correctly as covered or uncovered. Then we use union bound to bound the probability that all clusters are detected correctly. Recall that each partition returned by {\tt PartitionSample} has size at least ${|T_j|}\geq \frac{2^6k}{\eps}$ for $j\in [l]$. We first compute for one such partition $T_j$ and some center $c\in C$, the probability that majority of the queries $\OA^E(v,c)$ where $v\in T_j$ are wrong. Since each query answer is wrong independently with probability $q<1/2$, in expectation the number of wrong query answers would be $q|T_j|$. Using Chernoff bound, the probability that majority of the queries is wrong is at most $e^{-\frac{2^6 k}{3\eps} (1-\frac{1}{2q})^2}$. The probability that the majority of the queries is wrong for at least one center $c\in C$ is at most $ke^{-\frac{2^6 k}{3\eps} (1-\frac{1}{2q})^2}$. Again using union bound all clusters are detected correctly with probability at least $(1-k^2e^{-\frac{2^6 k}{3\eps} (1-\frac{1}{2q})^2}) \geq (1-\frac{1}{16k})$. \qed \end{proof}

With probability at least $(1-\frac{1}{8k})$, given a $D^2$-sample $S$, we can figure out the largest uncovered optimal cluster using the sub-routines {\tt PartitionSample} and {\tt IsCovered}. The analysis of the Algorithm \ref{fig:algo3} follows the analysis of Algorithm \ref{fig:algo}. For completeness, we compute the probability of success, and the query complexity of the algorithm. Note that $s$ in line (3.2) of the Algorithm \ref{fig:algo3} is chosen correctly with probability $(1-\frac{1}{4k})(1-\frac{1}{8k})$. The uniform sample in line (3.3) is chosen properly with probability $(1-\frac{1}{4k})(1-\frac{1}{8k})$. Since, given the uniform sample, success probability using Inaba's lemma is at least $(1-\frac{1}{4k})$, overall the probability of success becomes $(1-\frac{1}{k})$. For query complexity, we observe that {\tt PartitionSample} makes $O(\frac{k^6}{\eps^8})$ same-cluster queries to oracle $\OA^E$, query complexity of {\tt IsCovered} is $O(\frac{k^4}{\eps^4})$. Since {\tt PartitionSample} is called at most $k$ times, total query complexity would be $O(\frac{k^7}{\eps^8})$.
Note that these are bounds for dataset that satisfies $(k, \veps)$-irreducibility condition.
For general dataset, we will use $O(\veps/k)$ as the error parameter. 
This causes the number of same-cluster queries to be $O(k^{15}/\veps^8)$.

\section{Conclusion and Open Problems}
This work explored the power of the SSAC framework defined by Ashtiani \etal\cite{akd16} in the approximation algorithms domain.
We showed how same-cluster queries allowed us to convert the popular $k$-means++ seeding algorithm into an algorithm that gives constant approximation for the $k$-means problem instead of the $O(\log{k})$ approximation guarantee of $k$-means++ in the absence of such queries.
Furthermore, we obtained an efficient $(1 + \veps)$-approximation algorithm for the $k$-means problem within the SSAC framework.
This is interesting because it is known that such an efficient algorithm is not possible in the classical model unless $\mathsf{P} = \mathsf{NP}$.

Our results encourages us to formulate similar query models for other hard problems.
If the query model is reasonable (as is the SSAC framework for center-based clustering), then it may be worthwhile exploring its powers and limitations as it may be another way of circumventing the hardness of the problem.
For instance, the problem closest to center-based clustering problems such as $k$-means is the {\em correlation clustering} problem.
The query model for this problem may be similar to the SSAC framework.
It will be interesting to see if same-cluster queries allows us to design efficient algorithms and approximation algorithms for the correlation clustering problem for which hardness results similar to that of $k$-means is known.

\bibliographystyle{alpha}
\bibliography{paper}

\end{document}